\documentclass{svjour3}                     
\smartqed  
\usepackage[dvips]{color}
\usepackage{epsf}
\usepackage{times}
\usepackage{epsfig}
\usepackage{graphicx}
\usepackage{url}
\usepackage{amsmath}
\usepackage{amssymb}
\usepackage{amsxtra}
\usepackage{algorithmic}
\usepackage{algorithm}
\usepackage{multirow}
\usepackage{cite}
\usepackage{enumerate}
\usepackage{wrapfig}
\newtheorem{pro}{\bf Property}
\newtheorem{remarks}{\bf Remark}
\newlength{\aligntop}
\newlength{\alignbot}
\makeatletter
\renewenvironment{align}{%
  \vspace{\aligntop}
  \start@align\@ne\st@rredfalse\m@ne
}{%
  \math@cr \black@\totwidth@
  \egroup
  \ifingather@
    \restorealignstate@
    \egroup
    \nonumber
    \ifnum0=`{\fi\iffalse}\fi
  \else
    $$%
  \fi
  \ignorespacesafterend%
  \vspace{\alignbot}\par\noindent
}
\begin{document}

\title{{\Huge Distributed Coalition Formation Games for Secure Wireless Transmission}\thanks{This work was done, in part, during the stay of Walid Saad at the Coordinated Science Laboratory, University of Illinois at Urbana-Champaign and was supported by the Research Council of Norway through projects 183311/S10, 176773/S10, and 18778/V11 and by NSF grants CNS-0905556 and CNS-0910461. A preliminary version of this paper appears in the Proceedings of the 7th International Symposium on Modeling and Optimization in Mobile, Ad Hoc and Wireless Networks [18].}
}


\author{ Walid Saad \and Zhu Han \and Tamer Ba\c{s}ar \and M{\'e}rouane Debbah \and Are Hj{\o}rungnes}


\institute{W. Saad \and A. Hj{\o}rungnes \at
              UNIK - University Graduate Center\\
              University of Oslo, Norway\\
              \email{\{saad,arehj\}@unik.no}           
           \and
           Z. Han \at
              Electrical and Computer Engineering Department\\
              University of Houston, TX, USA\\
              \email{zhan2@mail.uh.edu}           
              \and
           M. Debbah \at
              Alcatel Lucent Chair\\
              SUPELEC, France\\
              \email{merouane.debbah@supelec.fr}           
              \and
           T. Ba\c{s}ar  \at
              Coordinated Science Laboratory\\
              University of Illinois at Urbana Champaign, USA\\
              \email{basar1@illinois.edu}           
              \and
}

\date{Received: date / Accepted: date\vspace{-0.7cm}}

\maketitle
\begin{abstract}
Cooperation among wireless nodes has been recently proposed for improving the physical layer~(PHY) security of wireless transmission in the presence of multiple eavesdroppers. While existing PHY security literature  answered the question ``what are the link-level \emph{secrecy rate} gains from cooperation?'', this paper attempts to answer the question of ``how to achieve those gains in a practical decentralized wireless network and in the presence of a cost for information exchange?''. For this purpose, we model the PHY security cooperation problem as a coalitional game with non-transferable utility and propose a distributed algorithm for coalition formation. Through the proposed algorithm, the wireless users can cooperate and self-organize into disjoint independent coalitions, while maximizing their secrecy rate taking into account the security costs during information exchange. We analyze the resulting coalitional structures for both decode-and-forward and amplify-and-forward cooperation and study how the users can adapt the network topology to environmental changes such as mobility. Through simulations, we assess the performance of the proposed algorithm and show that, by coalition formation using decode-and-forward, the average secrecy rate per user is increased of up to $25.3\%$ and $24.4\%$ (for a network with $45$~users) relative to the non-cooperative and amplify-and-forward cases, respectively.
\keywords{physical layer security \and coalitional games\and game theory \and secure communication}\vspace{-0.5cm}
\end{abstract}
\section{Introduction}\vspace{-0.4cm}
With the recent emergence of ad hoc and decentralized networks, higher-layer security techniques such as encryption have become hard to implement. This led to an increased attention on studying the ability of the physical layer~(PHY)  to provide secure wireless communication. The main idea is to exploit the wireless channel PHY characteristics such as fading or noise for improving the reliability of wireless transmission. This reliability is quantified by the rate of secret information sent from a wireless node to its destination in the presence of eavesdroppers, i.e., the so called \emph{secrecy rate}. The maximal achievable secrecy rate is referred to as the \emph{secrecy capacity}. The study of this security aspect began with the pioneering work of Wyner over the wire-tap channel \cite{WY00} and was followed up in \cite{WY01,WY02} for the scalar Gaussian wire-tap channel and the broadcast channel, respectively.

Recently, there has been a growing interest in carrying out these studies unto the wireless and the multi-user channels \cite{WI00,WI01,WI02,WI04,WI06,WI05,WI07}. For instance, in \cite{WI00} and \cite{WI01}, the authors study the secrecy capacity region for both the Gaussian and the fading broadcast channels and propose optimal power allocation strategies. In \cite{WI02}, the secrecy level in multiple access channels from a link-level perspective is studied. Further,  multiple antenna systems have been proposed in \cite{WI06} for ensuring a non-zero secrecy capacity. The work in \cite{WI05,WI07} presents a performance analysis for using cooperative beamforming (with no cost for cooperation), with decode-and-forward and amplify-and-forward relaying, to improve the secrecy rate of a single cluster consisting of one source node and a number of relays. Briefly, the majority of the existing literature is devoted to the information theoretic analysis of link-level performance gains of secure communications with no information exchange cost, notably when a source node cooperate with some relays as in \cite{WI05,WI07}. While this literature studied the performance of some cooperative schemes, no work seems to have investigated how a number of users, each with its own data, can interact and cooperate at network-wide level to improve their secrecy rate.

 The main contribution of this work is to propose distributed cooperation strategies, through coalitional game theory \cite{WSTUT}, which allow to study the interactions between a network of users that seek to secure their communication in the presence of multiple eavesdroppers. Another major contribution is to study the impact on the network topology and dynamics of the inherent tradeoff that exists between the PHY security cooperation gains in terms of secrecy rate and the information exchange costs. In other words, while the earlier work answered the question ``what are the secrecy rate gains from cooperation?'', here, we seek to answer the question of ``how to achieve those gains in a practical decentralized wireless network and in the presence of a cost for information exchange?''. We model the problem as a non-transferable coalitional game and propose a distributed algorithm for autonomous coalition formation based on well suited concepts from cooperative games. Through the proposed algorithm, each user autonomously decides to form or break a coalition for maximizing its utility in terms of secrecy rate while accounting for the loss of secrecy rate during information exchange. We show that independent disjoint coalitions form in the network, due to the cooperation cost, and we study their properties for both the decode-and-forward and amplify-and-forward cooperation models.Simulation results show that, by coalition formation using decode-and-forward, the average secrecy rate per user is increased of up to $25.3\%$ and $24.4\%$ relative to the non-cooperative and amplify-and-forward cases, respectively. Further, the results show how the users can self-organize and adapt the topology to mobility.

The rest of this paper is organized as follows: Section~\ref{sec:systemmodel} presents
the system model. Section~\ref{sec:gf} presents the game formulation and properties. In Section~\ref{sec:coalform} we devise the coalition formation algorithm. Simulation results are presented and analyzed in Section \ref{sec:sim}. Finally, conclusions are drawn in
Section \ref{sec:conc}.\vspace{-0.7cm}
\section{System Model}\vspace{-0.4cm}
\label{sec:systemmodel}
Consider a network having $N$  transmitters
(e.g.~mobile users) sending data to $M$ receivers (destinations)
in the presence of $K$ eavesdroppers that seek to tap into the transmission of the users. Users, receivers and eavesdroppers are unidirectional-single-antenna nodes. We define $\mathcal{N}=\{1,\ldots,N\}$, $\mathcal{M}=\{1,\ldots,M\}$ and $\mathcal{K}=\{1,\ldots,K\}$ as the sets of users, destinations, and eavesdroppers, respectively. In this work, we consider only the case of multiple eavesdroppers, hence, we have $K > 1$. Furthermore, let $h_{i,m_i}$ denote the complex baseband channel gain between user $i \in \mathcal{N}$ and its destination $m_i \in \mathcal{M}$ and $g_{i,k}$ denote the channel gain between user $i \in \mathcal{N}$ and eavesdropper $k \in \mathcal{K}$. We consider a line of sight channel model with $h_{i,m_i}=d_{i,m_i}^{-\frac{\mu}{2}}e^{j\phi_{i,m_i}}$ with $d_{i,m_i}$ the distance between user $i$ and its destination $m_i$, $\mu$ the pathloss exponent, and $\phi_{i,m_i}$ the phase offset. A similar model is used for the user-eavesdropper channel. Note that other channel models can also be accommodated.

Further, we consider a TDMA transmission, whereby, in a non-cooperative manner, each user occupies a single time slot. Within a single slot, the amount of reliable information transmitted from the user $i$ occupying the slot to its destination $m_i$ is quantified through the \emph{secrecy rate} $C_{i,m_i}$ defined as follows \cite{WI00}:
\begin{equation}\label{eq:sec}
C_{i,m_i} = \left( C^{d}_{i,m_i} - \max_{1\le k \le K}{C^{e}_{i,k}}\right)^{+},
\end{equation}
\noindent where $C^{d}_{i,m_i}$ is the capacity for the transmission between user $i$ and its destination $m_i \in \mathcal{M}$, $C^{e}_{i,k}$ is the capacity of user $i$ at the eavesdropper $k \in \mathcal{K}$, and $a^{+} \triangleq \max{(a,0)}$. Note that the secrecy rate in (\ref{eq:sec}) is shown to be achievable  in \cite{ACH} using Gaussian inputs.
\begin{figure}[!t]
\begin{center}
\includegraphics[width=0.6\textwidth]{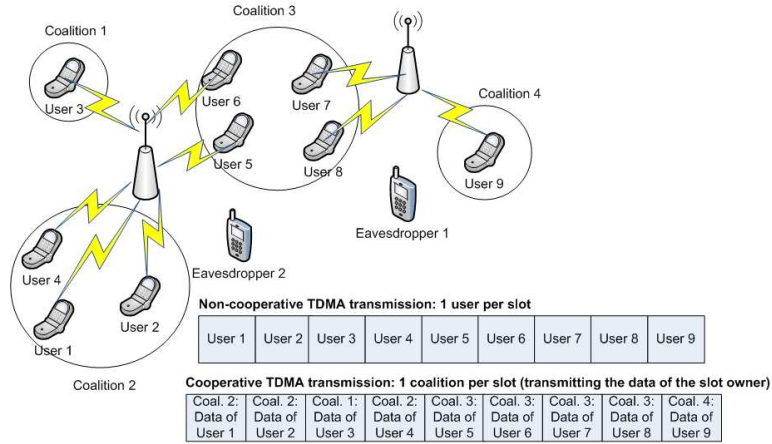}
\end{center}\vspace{-0.6cm}
\caption{System model for physical layer security coalitional game.}\vspace{-6mm}%
\label{f: system_model}%
\end{figure}

In a non-cooperative approach, due to the broadcast nature of the wireless channel, the transmission of the users can be overheard by the eavesdroppers which reduces their secrecy rate as clearly expressed in (\ref{eq:sec}). For improving their performance and increasing their secrecy rate
, the users can collaborate by forming coalitions. Within every coalition, the users can utilize collaborative beamforming techniques for improving their secrecy rates. In this context, every user $i$ member of a coalition $S$ can cooperate with its partners in $S$ by dividing its slot into two durations:
\begin{enumerate}
\item In the first duration, user $i$ broadcasts its data to the other members of coalition $S$.
    \item In the second duration, coalition $S$ performs collaborative beamforming. Thus, all the members of coalition $S$ relay a weighted version of user $i$'s signal to its destination.
 \end{enumerate}
 Although finding an optimal cooperation scheme that maximizes the secrecy rate is quite complex \cite{WI05}, one approach for cooperation is to null the signal at the eavesdroppers, i.e., impose $C^{e}_{i,k}=0,\forall k \in \mathcal{K}$, hence, improving their secrecy rate as compared to the non-cooperative rate in (\ref{eq:sec}) \cite{WI05}. Each coalition $S \subseteq \mathcal{N}$ that forms in the network is able to transmit within all the time slots previously held by its users. Thus, in the presence of cooperating coalitions, the TDMA system schedules one coalition per time slot. During a given slot, the coalition acts as a single entity for transmitting the data of the user that owns the slot. Fig.~\ref{f: system_model} shows an illustration of this model for $N=9$~users, $M=2$~destinations, and $K=2$~eavesdroppers. 

Furthermore, we define a fixed transmit power \emph{per time slot} $\tilde{P}$ which constrains \emph{all the users} that are transmitting within a given slot. In a non-cooperative manner, this power constraint applies to the single user occupying the slot, while in a cooperative manner this \emph{same} power constraint applies to the entire coalition occupying the slot. Such a power assumption is typical in TDMA systems comprising mobile users and is a direct result of ergodicity and the time varying user locations \cite{GC01,GC00,WS00}. For every coalition $S$, during the time slot owned by user $i\in S$, user $i$ utilizes a portion of the available power $\tilde{P}$ for information exchange (first stage) while the remaining portion $P_i^{S}$ is used by the coalition $S$ to transmit the actual data to the destination $m_i$ of user $i$ (second stage). For information exchange, user $i\in S$ can broadcast its information to the farthest user $\hat{i} \in S$, by doing so all the other members of $S$ can also obtain the information due to the broadcast nature of the wireless channel. This information exchange incurs a power cost  $\bar{P}_{i,\hat{i}}$ given by
\begin{equation}\label{eq:scost}
\bar{P}_{i,\hat{i}} = \frac{\nu_0 \cdot \sigma^{2}}{|q_{i,\hat{i}}|^2} ,
\end{equation}
where $\nu_0$ is a target average signal-to-noise ratio (SNR) for information exchange, $\sigma^{2}$ is the noise variance and $q_{i,\hat{i}}$ is the channel gain between users $i$ and $\hat{i}$. The remaining power that coalition $S$ utilizes for the transmission of the data of user $i$ during the remaining time of this user's slot is
\begin{equation}\label{eq:pcost}
P_i^{S}=(\tilde{P} - \bar{P}_{i,\hat{i}})^{+}.
\end{equation}

For every coalition $S$, during the transmission of the data of user $i$ to its destination, the coalition members can cooperate, using either decode-and-forward~(DF) or amplify-and-forward~(AF), and, hence, weigh their signals in a way to \emph{completely null} the signal at the eavesdroppers. In DF, the coalition members that are acting as relays decode the received signal in the information exchange phase, then re-encode it before performing beamforming. In contrast, for AF, coalition members that are acting as relays perform beamforming by weighing the noisy version of the received signal in the information exchange phase. For any coalition $S$ the signal weights and the ``user-destination'' channels are represented by the $|S|\times 1$  vectors $\boldsymbol{w}_S=[w_{i_1},\ldots,w_{i_{|S|}}]^{H}$ and  $\boldsymbol{h}_S = [h_{i_1,m_1},\ldots,h_{i_{|S|},m_{|S|}}]^{H}$, respectively.
 By nulling the signals at the eavesdropper through DF cooperation within coalition $S$, the secrecy rate achieved by user $i\in S$ at its destination $m_i$ during user $i$'s time slot becomes \cite[Eq. (14)]{WI05}
\begin{equation}\label{eq:seccoop}
C^{S,\textrm{DF}}_{i,m_i} = \frac{1}{2}\log_2{\left(1+\frac{(\boldsymbol{w}_S^{*,\textrm{DF}})^{H}\boldsymbol{R}_S\boldsymbol{w}_S^{*,DF}}{\sigma^2}\right)},
\end{equation}
where $\boldsymbol{R}_S =\boldsymbol{h}_S\boldsymbol{h}_S^{H}$, $\sigma^2$ is the noise variance, and  $\boldsymbol{w}^{*,\textrm{DF}}_{S}$ is the weight vector that maximizes the secrecy rate while nulling the signal at the eavesdropper with DF cooperation and can be found using \cite[Eq.~(20)]{WI05}. In (\ref{eq:seccoop}), the factor $\frac{1}{2}$ accounts for the fact that half of the slot of user $i$ is reserved for information exchange. 

For AF, we define, during the transmission slot of a user $i \in S$ member of a coalition $S$, the $|S|\times 1$ vector $\boldsymbol{a}_S^{i}$ with every element $a_{S,j}^{i} = \sqrt{\bar{P}_{i,\hat{i}}}q_{i,j}h_{j,m_j}, \ \forall j \neq i$ ($q_{i,j}$ is the channel between users $i$ and $j$ and $\bar{P}_{i,\hat{i}}$ is the power used by user $i$ for information exchange as per (\ref{eq:scost})) and $a_{S,i}^{i}=\sqrt{\bar{P}_{i,\hat{i}}}h_{i,m_i}$ and the $|S|\times|S|$ diagonal matrix $\boldsymbol{U}_S^{i}$ with every diagonal element $u_{S,j,j}^{i}=|h_{j,m_j}|^2 \ \forall j \neq i$ and $u_{S,i,i}^{i}=0$. Given these definitions and by nulling the signals at the eavesdropper through AF cooperation within coalition $S$, the secrecy rate achieved by user $i\in S$ at its destination $m_{i}$ during user $i$'s time slot becomes \cite[Eq. (3)]{WI07}
\begin{equation}\label{eq:seccoopAF}
C^{S,\textrm{AF}}_{i,m_i} = \frac{1}{2}\log_2{\left(1+\frac{(\boldsymbol{w}_S^{*,\textrm{AF}})^{H}\boldsymbol{R}_a\boldsymbol{w}_S^{*,\textrm{AF}}}{(\boldsymbol{w}_S^{*,\textrm{AF}})^{H}\boldsymbol{U}_S^i\boldsymbol{w}_S^{*,\textrm{AF}}+1)\sigma^2}\right)},
\end{equation}
where $\boldsymbol{R}_a =\boldsymbol{a}_S^i(\boldsymbol{a}_S^i)^{H}$,  and  $\boldsymbol{w}^{*,\textrm{AF}}_{S}$ is the weight vector that maximizes the secrecy rate while nulling the signal at the eavesdropper with AF cooperation and can be found using \cite[Eqs.(14)-(15)]{WI07}. Note that for AF, as seen in (\ref{eq:seccoopAF}) there is a stronger dependence on the channels (through the matrix $\boldsymbol{R}_a$) between the cooperating users in both the first and second phase of cooperation, unlike in DF, where this dependence is solely through the power in (\ref{eq:scost}) during the information exchange phase. Further, for AF, as the cooperating users amplify a noisy version of the signal, the noise is also amplified, which can reduce the cooperation gains, as seen through the term $(\boldsymbol{w}_S^{*,\textrm{AF}})^{H}\boldsymbol{U}_S^i\boldsymbol{w}_S^{*,\textrm{AF}}$.

Further, it must be stressed that, although the models for AF and DF cooperation in (\ref{eq:seccoop}) and (\ref{eq:seccoopAF}) are inspired from \cite{WI05,WI07}, our work and contribution differ significantly from \cite{WI05,WI07}. While the work in \cite{WI05,WI07} is solely dedicated to finding the optimal weights in (\ref{eq:seccoop}) and (\ref{eq:seccoopAF}), and presenting a link-level performance analysis for a single cluster of neighboring nodes with no cost for cooperation, our work seeks to perform a network-level analysis by modeling the interactions among a network of users that seek to cooperate,  in order to improve their performance, using either the DF or AF protocols in the presence of costs for information exchange. Hence, the main focus of this paper is modeling the user's behavior, studying the network dynamics and topology, and analyzing the network-level aspects of cooperation in PHY security problems. In this regard, the remainder of this paper is devoted to investigate how a network of users can cooperate, through the protocols described in this section, and improve the security of their wireless transmission, i.e., their secrecy rate.

Finally, note that, in this paper, we assume that the users have perfect knowledge of the channels to the eavesdroppers which is an assumption commonly used in most PHY security related literature, and as explained in \cite{CSI00} this channel information  can be obtained by the users through a constant monitoring of the behavior of the eavesdroppers. Alternatively, the eavesdroppers in this work can also be seen as areas where the transmitters suspect the presence of malicious eavesdropping nodes and, hence, need to secure these locations. Hence, our current analysis can serve as an upper bound for future work where the analysis pertaining to the case where the eavesdroppers and their locations are not known will be tackled (in that case although the cooperation model needs to be modified, the PHY security coalitional game model presented in the following sections can be readily applied).\vspace{-0.5cm}
\section{Physical Layer Security as A Coalitional Game}\label{sec:gf}\vspace{-0.3cm}
The proposed PHY security problem can be modeled as a $(\mathcal{N},V)$ coalitional game with a non-transferable utility \cite{Game_theory2,WSTUT} where $V$ is a mapping such that for every coalition $S \subseteq \mathcal{N}$, $V(S)$ is a closed convex subset of $\mathbb{R}^{|S|}$  that contains the payoff vectors that players in $S$ can achieve. Thus, given a coalition $S$ and denoting by $\phi_i(S)$ the payoff of user $i \in S$ during its time slot, we define the coalitional value set, i.e., the mapping $V$ as follows
\begin{align}\label{eq:util}
V(S)=\{\boldsymbol{\phi}(S) \in \mathbb{R}^{|S|}|\ \ \forall i \in S \ \phi_i(S) = (v_i(S)-c_i(S))^{+}\nonumber\\ \textrm{ if } P_i^{S} >0, \textrm{ and } \phi_i(S) = -\infty\textrm{ otherwise.}\},
\end{align}
where $v_i(S)=C^{S}_{i,m_i} $ is the gain in terms of secrecy rate for user $i \in S$ given by (\ref{eq:seccoop}) while taking into account the available power $P_i^{S}$ in (\ref{eq:pcost}) and $c_i(S)$ is a secrecy cost function that captures the loss for user $i \in S$, in terms of secrecy rate, that occurs during information exchange. Note that, when \emph{all} the power is spent for information exchange, the payoff $\phi_i(S)$ of user $i$ is set to $-\infty$ since, in this case, the user has clearly no interest in cooperating.

With regard to the secrecy cost function $c_i(S)$, when a user $i \in S$ sends its information to the farthest user $\hat{i} \in S$ using a power level $\bar{P}_{i,\hat{i}}$, the eavesdroppers can overhear the transmission. This security loss is quantified by the capacity at the eavesdroppers resulting from the information exchange and which, for a particular eavedropper $k \in \mathcal{K}$, is given by
$\hat{C}^{e}_{i,k} = \frac{1}{2} \log{(1+\frac{\bar{P}_{i,\hat{i}} \cdot |g_{i,k}|^{2}}{\sigma^2})}$
and the cost function $c(S)$ can be defined as
\begin{equation}\label{eq:cost2}
c_i(S)=\max{(\hat{C}^{e}_{i,1},\ldots,\hat{C}^{e}_{i,K})}.
\end{equation}

In general, coalitional game based problems seek to characterize the properties and stability of the grand coalition of all players since it is generally assumed that the grand coalition maximizes the utilities of the players \cite{Game_theory2}. In our case, although cooperation improves the secrecy rate as per (\ref{eq:util}) for the users in the TDMA network; the utility in (\ref{eq:util}) also accounts for two types of cooperation costs:(i)- The fraction of power spent for information exchange as per (\ref{eq:pcost}) and, (ii) the secrecy loss during information exchange as per (\ref{eq:cost2}) which can  strongly limit the cooperation gains. Therefore, for the proposed $(\mathcal{N},v)$ coalitional game we have:
\begin{pro}\label{prop:2}
For the proposed $(\mathcal{N},V)$ coalitional game, the grand coalition of all the users \emph{seldom} forms due to the various costs for information exchange. Instead, disjoint independent coalitions will form in the network.
\end{pro}
\begin{proof}
The proof is found in \cite[Property~2]{WIOPT}.\vspace{-0.3cm}
\end{proof}

Due to this property, traditional solution concepts for coalitional games, such as the core \cite{Game_theory2}, may not be applicable \cite{WSTUT}. In fact, in order for the core to exist, as a solution concept, a coalitional game must ensure that the grand coalition, i.e., the coalition of all players will form. However, as seen in Figure~\ref{f: system_model} and corroborated by Property~\ref{prop:2}, in general, due to the cost for coalition formation, the grand coalition will not form. Instead, independent and disjoint coalitions appear in the network as a result of the collaborative beamforming process. In this regard, the proposed game is classified as a \emph{coalition formation game} \cite{WSTUT}, and the objective is to find the coalitional structure that will form in the network, instead of finding only a solution concept, such as the core, which aims mainly at stabilizing the grand coalition.

Furthermore, for the proposed $(\mathcal{N},V)$ coalition formation game, a constraint on the coalition size, imposed by the nature of the cooperation protocol exists as follows:
\begin{remarks}
For the proposed $(\mathcal{N},V)$ coalition formation game,  the size of any coalition $S \subseteq \mathcal{N}$ that will form in the network must satisfy $|S| > K$ for both DF and AF cooperation.
\end{remarks}

This is a direct result of the fact that, for nulling $K$ eavesdroppers, at least $K+1$ users must cooperate, otherwise, no weight vector can be found to maximize the secrecy rate while nulling the signal at the eavesdroppers.\vspace{-0.2cm}

\section{Distributed Coalition Formation Algorithm}\label{sec:coalform}\vspace{-0.4cm}
\subsection{Coalition Formation Algorithm}\label{sec:mergeandsplit}\vspace{-0.3cm}
Coalition formation has recently attracted increased attention in game theory \cite{WSTUT,CF04,KA01}. The goal of coalition formation games is to find algorithms for characterizing the coalitional structures that form in a network where the grand coalition is not optimal. For constructing a coalition formation process suitable to the proposed $(\mathcal{N},V)$ PHY security cooperative game, we require the following definitions \cite{WSTUT,KA01}
 \begin{definition}
A \emph{collection} of coalitions, denoted by $\mathcal{S}$, is defined as the set $\mathcal{S} = \{S_{1},\ldots,S_{l}\}$ of mutually disjoint coalitions $S_{i} \subset \mathcal{N}$.  In other words, a collection is any arbitrary group of disjoint coalitions $S_i$ of $\mathcal{N}$ not necessarily spanning all players of $ \mathcal{N}$. If the collection spans \emph{all} the players of $ \mathcal{N}$; that is $\bigcup_{j=1}^{l} S_j =  \mathcal{N}$, the collection is a \emph{partition} of $\mathcal{N}$.
 \end{definition}
\begin{definition}
A preference operator or \emph{comparison relation} $\rhd$ is an order defined for comparing two collections $\mathcal{R} = \{R_{1},\ldots,R_{l}\}$ and
$\mathcal{S} = \{S_{1},\ldots,S_{p}\}$ that are partitions of the same subset $\mathcal{A} \subseteq \mathcal{N}$ (i.e.~same players in $\mathcal{R}$ and $\mathcal{S}$). Therefore, $\mathcal{R} \rhd \mathcal{S}$ implies that the way $\mathcal{R}$ partitions $\mathcal{A}$ is preferred to the way $\mathcal{S}$ partitions $\mathcal{A}$.
\end{definition}

For the proposed PHY security coalition formation game, an individual value order, i.e. an order which compares the individual payoffs of the users, is needed due to the non-transferable utility of the game. For this purpose, for the proposed game, we utilize the following order for defining the preferences of the users

%
\begin{definition}
Consider two collections $\mathcal{R} = \{R_{1},\ldots,R_{l}\}$ and
$\mathcal{S} = \{S_{1},\ldots,S_{m}\}$ that are partitions of the same subset $\mathcal{A} \subseteq \mathcal{N}$ (same players in $\mathcal{R}$ and $\mathcal{S}$). For a collection $\mathcal{R} = \{R_{1},\ldots,R_{l}\}$, let the utility of a player $j$ in a coalition $R_j \in \mathcal{R}$ be denoted by $\Phi_j(\mathcal{R})=\phi_j(R_j) \in V(R_j)$. $\mathcal{R}$ is preferred over $\mathcal{S}$ by \emph{Pareto order}, written as $\mathcal{R} \rhd \mathcal{S}$, iff
\setlength{\aligntop}{-0.2em}
\setlength{\alignbot}{-2\baselineskip}
\begin{align}\label{eq:Pareto}
\mathcal{R} \rhd \mathcal{S} \Longleftrightarrow \{\Phi_{j}(\mathcal{R}) \ge
\Phi_{j}(\mathcal{S}) \  \forall \ j \in \mathcal{R},\mathcal{S}\},\nonumber
\\\textrm{ with \emph{at
least one strict inequality} ($>$) for a player } k. \nonumber
\end{align}
\setlength{\aligntop}{-0.2em}
\setlength{\alignbot}{-0.8\baselineskip}
\vspace{-0.5cm}
\end{definition}

In other words, a collection is preferred by the players over another collection, if at least one player is able to improve its payoff without hurting the other players. Subsequently, for performing autonomous coalition formation between the users in the proposed PHY security game, we construct a distributed algorithm based on two simple rules denoted as ``merge'' and ``split'' \cite{WSTUT,KA01} defined as follows.
\begin{definition}
\textbf{Merge Rule -} Merge any set of coalitions
$\{S_{1},\ldots,S_{l}\}$ whenever the merged form is preferred by the players, i.e., where
$\{\bigcup_{j=1}^{l}S_{j}\} \rhd \{S_{1},\ldots,S_{l}\}$, therefore, {$\{S_{1},\ldots,S_{l}\}
\rightarrow \{\bigcup_{j=1}^{l}S_{j}\}$}.
\end{definition}
\begin{definition}
\textbf{Split Rule -} Split any coalition
$\bigcup_{j=1}^{l}S_{j}$ whenever a split form is preferred by the players, i.e., where
$\{S_{1},\ldots,S_{l}\} \rhd \{\bigcup_{j=1}^{l}S_{j}\},$
 thus, $\{\bigcup_{j=1}^{l}S_{j}\}
\rightarrow \{S_{1},\ldots,S_{l}\}$.
\end{definition}
 \begin{table}[!t]
  \caption{\vspace*{-0.4em}One round of the proposed PHY security coalition formation algorithm}\vspace*{-0.2em}
    \begin{tabular}{p{10cm}}
      \hline
      \textbf{Initial State} \vspace*{.2em} \\
      \hspace*{1em} The network is partitioned by $\mathcal{T}=\{T_1,\ldots,T_k\}$ (At the beginning\\
      \hspace*{1em}of all time $\mathcal{T}$ = $\mathcal{N}$ = $\{1,\ldots,N\}$ with non-cooperative users). \vspace*{.3em}\\
\textbf{Three phases in each round of the coalition formation algorithm} \vspace*{.2em}\\
\hspace*{1em}\emph{Phase 1 - Neighbor Discovery:}   \vspace*{.1em}\\
\hspace*{2em}a) Each coalition surveys its neighborhood for candidate partners.\vspace*{.1em}\\
\hspace*{2em}b) For every coalition $T_i$, the candidate partners lie in the area \vspace*{.1em}\\
\hspace*{2em}represented by the intersection of  $|T_i|$ circles with centers $j \in T_i$\vspace*{.1em}\\
\hspace*{2em}and radii determined by the distance where the power for \vspace*{.1em}\\
\hspace*{2em}information exchange
does not exceed $\tilde{P}$ for any user\vspace*{.1em}\\
\hspace*{2em}(easily computed through (\ref{eq:scost})).\vspace*{.1em}\\
\hspace*{1em}\emph{Phase 2 - Adaptive Coalition Formation:}   \vspace*{.1em}\\
\hspace*{2em}In this phase, coalition formation using merge-and-split occurs. \vspace*{.2em}\\
\hspace*{3em}\textbf{repeat}\vspace*{.2em}\\
\hspace*{4em}a) $\mathcal{F}$ = Merge($\mathcal{T}$); coalitions in $\mathcal{T}$ decide to merge based on\vspace*{.2em}\\
\hspace*{4em}the algorithm of Section~\ref{sec:mergeandsplit}.\vspace*{.2em}\\ 
\hspace*{4em}b) $\mathcal{T}$ = Split($\mathcal{F}$); coalitions in $\mathcal{F}$ decide to split based on \vspace*{.2em}\\
\hspace*{4em}the Pareto order.\\
\hspace*{3em}\textbf{until} merge-and-split terminates.\vspace*{.2em}\\
\hspace*{1em}\emph{Phase 3 - Secure Transmission:}   \vspace*{.1em}\\
\hspace*{2em}Each coalition's users exchange their information
 and transmit\vspace*{.1em}\\
\hspace*{2em}their data within their allotted slots.\vspace*{.1em}\\
\textbf{The above three phases are repeated periodically during the network operation, allowing a topology that is adaptive to environmental changes such as mobility. } \vspace*{.5em}\\
   \hline
    \end{tabular}\label{tab:alg1}
    \vspace{-0.5cm}
\end{table}

Using the above rules, multiple coalitions can merge into a larger coalition if merging yields a preferred collection based on the Pareto order. This implies that a group of users can agree to form a larger coalition, if at least one of the users improves its payoff without decreasing the utilities of any of the other users. Similarly, an existing coalition can decide to split into smaller coalitions if splitting yields a preferred collection by Pareto order. The rationale behind these rules is that, once the users agree to sign a merge agreement, this agreement can only be broken if all the users approve. This is a family of coalition formation games known as coalition formation games with partially reversible agreements \cite{CF04}. Using the rules of merge and split is highly suitable for the proposed PHY security game due to many reasons. For instance, each merge or split decision can be taken in a distributed manner by each individual user or by each already formed coalition. Further, it is shown in \cite{KA01} that any arbitrary iteration of merge and split rules terminates, hence, these rules can be used as building blocks in a coalition formation process for the PHY security game. 

Accordingly, for the proposed PHY security game, we construct a coalition formation algorithm based on merge-and-split and divided into three phases: Neighbor discovery, adaptive coalition formation, and transmission. In the neighbor discovery phase (Phase~1), each coalition (or user) surveys its environment in order to find possible cooperation candidates. For a coalition $S_k$ the area that is surveyed for discovery is the intersection of  $|S_k|$ circles, centered at the coalition members with each circle's radius given by the maximum distance $\bar{r}_i$ (for the circle centered at $i \in S_k$) within which the power cost for user $i$ as given by (\ref{eq:scost}) does not exceed the total available power $\tilde{P}$. This area is determined by the fact that, if a number of coalitions $\{S_1,\ldots,S_m\}$ attempt to merge into a new coalition $G = \cup_{i=1}^{m}S_i$ which contains a member $i \in G$ such that the power for information exchange needed by $i$ exceeds $\tilde{P}$, then the payoff of $i$ goes to $-\infty$ as per (\ref{eq:util}) and the Pareto order can never be verified. Clearly, as the number of users in a coalition increases, the number of circles increases, reducing the area where possible cooperation partners can be found. This implies that, as the size of a coalition grows, the possibility of adding new users decreases, and, hence, the complexity of performing merge also decreases.

Following Phase~1, the adaptive coalition formation phase~(Phase~2) begins, whereby the users interact for assessing whether to form new coalitions with their neighbors or whether to break their current coalition. For this purpose, an iteration of sequential merge-and-split rules occurs in the network, whereby each coalition decides to merge (or split) depending on the utility improvement that merging (or splitting) yields. 
Starting from an initial network partition $\mathcal{T} =\{T_1,\ldots,T_l\}$ of $\mathcal{N}$, any random coalition (individual user) can start with the merge process. The coalition  $T_i \in  \mathcal{T}$ which debuts the merge process starts by enumerating, sequentially, the possible coalitions, of size greater than $K$ (Remark~1), that it can form with the neighbors that were discovered in Phase~1. On one hand, if a new coalition $\tilde{T}_i$ which is preferred by the users through Pareto order is identified, this coalition will form by a merge agreement of all its members. Hence,  the merge ends by a final merged coalition $T_i^{\text{final}}$ composed of $T_i$ and one or several of coalitions in its vicinity.  On the other hand, if $T_i$ is unable to merge with any of the discovered partners, it ends its search and $T_i^{\text{final}} = T_i$.

The algorithm is repeated for the remaining $T_i \in \mathcal{T}$ until all the coalitions have made their merge decisions, resulting in a final partition $\mathcal{F}$. Following the merge process, the coalitions in the resulting partition $\mathcal{F}$ are next subject to split operations, if any is possible. In the proposed PHY security problem, the coalitions are only interested in splitting into structures that include either singleton users or coalitions of size larger than $K$ or both (Remark~1). Similar to merge, the split is a local decision to each coalition. An iteration consisting of multiple successive merge-and-split operations is repeated until it terminates. The termination of an iteration of merge and split rules is guaranteed as shown in \cite{KA01}. It must be stressed that the merge or split decisions can be taken in a distributed way by the users/coalitions without relying on any centralized entity.

In the final transmission phase (Phase~3), the coalitions exchange their information and begin their secure transmission towards their corresponding destinations, in a TDMA manner, one coalition per slot. Every slot is owned by a user who transmits its data with the help of its coalition partners, if that user belongs to a coalition. Hence, in this phase, the user perform the actual beamforming, while transmitting the data of every user within its corresponding slot. Each run of the proposed algorithm consists of these three phases, and is summarized in Table~\ref{tab:alg1}. As time evolves and the users, eavesdroppers and destinations move (or new users or eavesdroppers enter/leave the network), the users can autonomously self-organize and adapt the network's topology through appropriate merge-and-split decisions during Phase~2. This adaptation to environmental changes is ensured by enabling the users to run the adaptive coalition formation phase periodically in the network.


The proposed algorithm in Table~\ref{tab:alg1} can be implemented in a distributed manner. As the user can detect the strength of other users' uplink signals (through techniques similar to those used in the ad hoc routing discovery) \cite{Zhu_book}, nearby coalitions can be discovered in Phase~1 for potential cooperation. In fact, during Phase~1, each coalition in the network can easily work out the area within which candidates for merge can be found, as previously explained in this section. Once the neighbors are discovered, the coalitions can perform merge operations based on the Pareto order in Phase~2. The complexity of the merge operation can grow exponentially with the number of candidates with whom a user $i$ is able to merge (the number of coalitions in the neighboring area which is in general significantly smaller than $N$). As more coalitions form, the area within which candidates are found is smaller, and, hence, the merge complexity reduces. In addition, whenever a coalition finds a candidate to merge with, it automatically goes through with the merge operation, hence, avoiding the need for finding all possible merge forms and reducing further the complexity. Further, each formed coalition can also internally decides to split if its members find a split form by Pareto order. By using a control channel, the distributed users can coordinate and then cooperate using our model. \vspace{-0.7cm}


\subsection{Partition Stability}\vspace{-0.3cm}
The result of the proposed algorithm in Table~\ref{tab:alg1} is a
network partition composed of disjoint independent coalitions. The stability of this network partition can be investigated using the concept of a defection function \cite{KA01}.
\begin{definition}
A \emph{defection} function $\mathbb{D}$ is a function which
associates with each partition $\mathcal{T}$ of $\mathcal{N}$ a group of
collections in $\mathcal{N}$. A partition $\mathcal{T} = \{T_1,\ldots,T_l\}$ of $\mathcal{N}$ is \emph{$\mathbb{D}$-stable} if no group of players is interested in leaving $\mathcal{T}$ when the players who leave can only form the collections allowed by $\mathbb{D}$.
\end{definition}

We are interested in two defection functions \cite{WSTUT,KA01}. First, the $\mathbb{D}_{hp}$ function which associates with each partition $\mathcal{T}$ of $\mathcal{N}$ the group of all partitions of $\mathcal{N}$ that can form through merge or split and the $\mathbb{D}_{c}$ function which associates with each partition $\mathcal{T}$ of $\mathcal{N}$ the group of all collections in $\mathcal{N}$. This function allows any group of players to leave the partition $\mathcal{T}$ of $\mathcal{N}$ through \emph{any} operation and create an arbitrary \emph{collection} in
$\mathcal{N}$.  Two forms of stability stem from these definitions: $\mathbb{D}_{hp}$ stability and a stronger $\mathbb{D}_{c}$
stability. A partition $\mathcal{T}$ is $\mathbb{D}_{hp}$-stable, if no
player in $\mathcal{T}$ is interested in leaving $\mathcal{T}$ through
merge-and-split to form other partitions in $\mathcal{N}$; while a partition
$\mathcal{T}$ is $\mathbb{D}_{c}$-stable, if no player in $\mathcal{T}$ is
interested in leaving $\mathcal{T}$ through \emph{any} operation (not
necessarily merge or split) to form other collections in $\mathcal{N}$.

Hence, a partition is $\mathbb{D}_{hp}$-stable if no coalition has an incentive to split or merge. For instance, a partition $\mathcal{T}=\{T_1,\ldots,T_l\}$ is $\mathbb{D}_{hp}$-stable, if the following two necessary and sufficient conditions are met \cite{WSTUT,KA01} ($\ntriangleright$ is the non-preference operator, opposite of $\rhd$): (i)- For each $i\in \{1,\ldots,m\}$ and for each partition $ \{R_1,\ldots,R_m\}$ of $T_i \in \mathcal{T}$ we have  $\{R_1,\ldots,R_m\}\ntriangleright T_i$, and (ii)- For each $S \subseteq \{1,\ldots,l\}$ we have
  $\bigcup_{i\in S}T_i\ntriangleright \{T_i|i\in S\}$. Using this definition of $\mathbb{D}_{hp}$ stability, we have
\begin{theorem}
\vspace{-0.1cm} Every partition resulting from our proposed coalition formation algorithm is $\mathbb{D}_{hp}$-stable.\vspace{-0.22cm}
\end{theorem}
\begin{proof}
The proof is given in \cite[Theorem~1]{WIOPT}.\vspace{-0.3cm}
\end{proof}

Furthermore, a $\mathbb{D}_{c}$-stable partition $\mathcal{T}$ is characterized by being a strongly stable partition, which  satisfies the following properties: (i)- A $\mathbb{D}_{c}$-stable partition is $\mathbb{D}_{hp}$-stable, (ii)- A $\mathbb{D}_{c}$-stable partition is a \emph{unique} outcome of any iteration of merge-and-split and, (iii)- A $\mathbb{D}_{c}$-stable partition $\mathcal{T}$ is a
unique $\rhd$-maximal partition, that is for all partitions $ \mathcal{T}'
\neq \mathcal{T}$ of $\mathcal{N}$, $\mathcal{T} \rhd \mathcal{T}'$. In the case where $\rhd$ represents the Pareto order, this implies that the $\mathbb{D}_{c}$-stable
partition $\mathcal{T}$ is the partition that presents a \emph{Pareto optimal} utility distribution for all the players.

Clearly, it is desirable that the network self-organizes unto a $\mathbb{D}_{c}$-stable partition. However, the existence of a $\mathbb{D}_{c}$-stable partition is not always guaranteed \cite{KA01}. The  $\mathbb{D}_{c}$-stable partition $\mathcal{T} = \{T_{1},\ldots,T_{l} \}$ of the whole space $\mathcal{N}$ exists if a partition of $\mathcal{N}$ that verifies the following two necessary and sufficient conditions exists\cite{KA01}:
\begin{enumerate}
\item For each $i\in \{1,\ldots,l\}$ and each pair of disjoint \emph{coalitions}  $S_1$ and $S_2$ such that $\{S_1 \cup S_2\} \subseteq T_i$ we have $\{S_1 \cup S_2\} \rhd \{S_1,S_2\}$.
\item For the partition $\mathcal{T}=\{T_1,\ldots,T_l\}$ a coalition $G \subset \mathcal{N}$ formed of players belonging to different $T_i \in \mathcal{T}$ is $\mathcal{T}$-incompatible if for no
$i \in \{1,\ldots,l\}$ we have $G\subset T_i$.
\end{enumerate}

In summary, $\mathbb{D}_{c}$-stability requires that for  all $\mathcal{T}$-incompatible coalitions $\{G\}[\mathcal{T}] \rhd \{G\}$ where $\{G\}[\mathcal{T}] = \{G\cap T_i \ \forall \ i\in\{1,\ldots,l\}\}$ is the projection of coalition $G$ on $\mathcal{T}$. If no partition of $\mathcal{N}$ can satisfy these conditions, then no $\mathbb{D}_{c}$-stable partition of $\mathcal{N}$ exists. Nevertheless, we have

\begin{lemma}
For the proposed $(\mathcal{N},v)$ PHY security coalitional game, the proposed algorithm of Table~\ref{tab:alg1} converges to the optimal $\mathbb{D}_{c}$-stable partition, if such a partition exists. Otherwise, the final network partition is $\mathbb{D}_{hp}$-stable.\vspace{-0.22cm}
\end{lemma}
\begin{proof}
The proof is a consequence of Theorem~1 and the fact that the  $\mathbb{D}_{c}$-stable partition is a unique outcome of any merge-and-split iteration \cite{KA01} which is the case with any partition resulting from our algorithm.\vspace{-0.3cm}
\end{proof}

Moreover, for the proposed game, the existence of the $\mathbb{D}_{c}$-stable partition cannot be always guaranteed. For instance, for verifying the first condition for existence of the  $\mathbb{D}_{c}$-stable partition, the users that are members of each coalitions must verify the Pareto order through their utility given by (\ref{eq:util}). Similarly, for verifying the second condition of $\mathbb{D}_{c}$ stability,  users belonging to all $\mathcal{T}$-incompatible coalitions in the network must verify the Pareto order. Consequently, the existence of such a  $\mathbb{D}_{c}$-stable partition is strongly dependent on the location of the users and eavesdroppers through the individual utilities (secrecy capacities). Hence, the existence of the $\mathbb{D}_{c}$-stable partition  is closely tied to the location of the users and the eavesdroppers, which, in a practical ad hoc wireless network are generally random. However, the proposed algorithm will always guarantee convergence to this optimal  $\mathbb{D}_{c}$-stable partition when it exists as stated in Lemma~1. Whenever a $\mathbb{D}_{c}$-stable partition does not exist, the coalition structure resulting from the proposed algorithm will be $\mathbb{D}_{hp}$-stable (no coalition or individual user is able to merge or split any further).\vspace{-5.5mm}

\section{Simulation Results and analysis}\label{sec:sim}\vspace{-0.3cm}
For simulations, a square network of $2.5$~km $\times$ $2.5$~km is set up with the users, eavesdroppers, and destinations randomly deployed within this area\footnote{This general network setting simulates a broad range of network types ranging from ad hoc networks, to sensor networks, WLAN networks as well as broadband or cellular networks.}. 
 In this network, the users are always assigned to the closest destination, although other user-destination assignments can be used without any loss of generality. For all simulations, the number of destinations is taken as $M=2$. Further, the power constraint per slot is set to $\tilde{P}=10$~mW, the noise level is $-90$~dBm, and the SNR for information exchange is $\nu_0=10$~dB which implies a neighbor discovery circle radius of $1$~km per user. For the channel model, the propagation loss is set to $\mu=3$.  All statistical results are averaged over the random positions of the users, eavesdroppers and destinations.


\begin{figure}[!t]
\begin{minipage}[b]{0.45\linewidth}
\centering
\includegraphics[width=\textwidth]{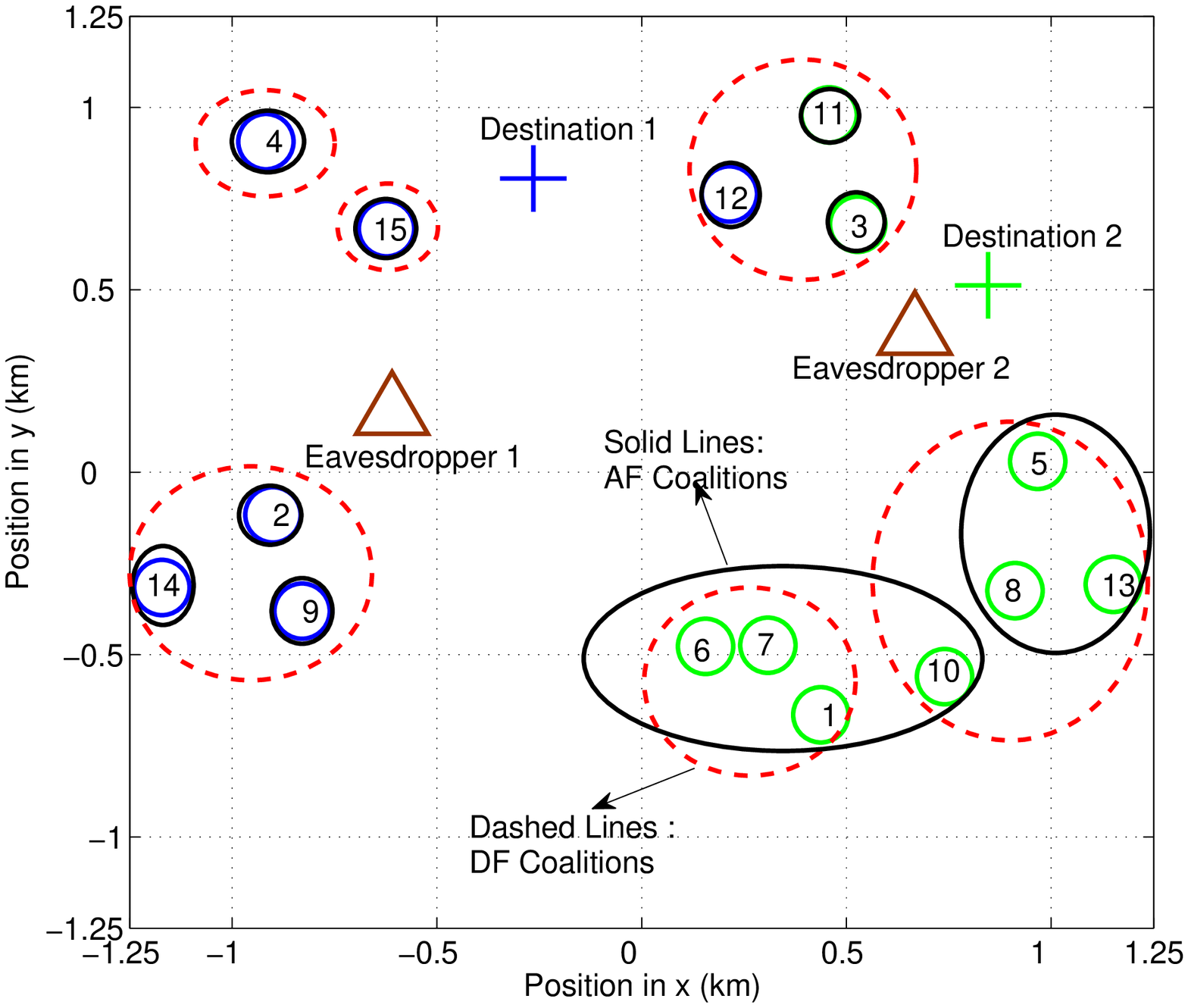}
\vspace{-0.6cm}
\caption{A snapshot of a coalitional structure resulting from our proposed coalition formation algorithm for a network with $N=15$~users, $M=2$~destinations and $K=2$~eavedroppers for DF (dashed lines) and AF (solid lines).}\label{fig:snapshot}\vspace{-0.6cm}
\end{minipage}
\hspace{0.5cm}
\begin{minipage}[b]{0.45\linewidth}
\centering
\includegraphics[width=\textwidth]{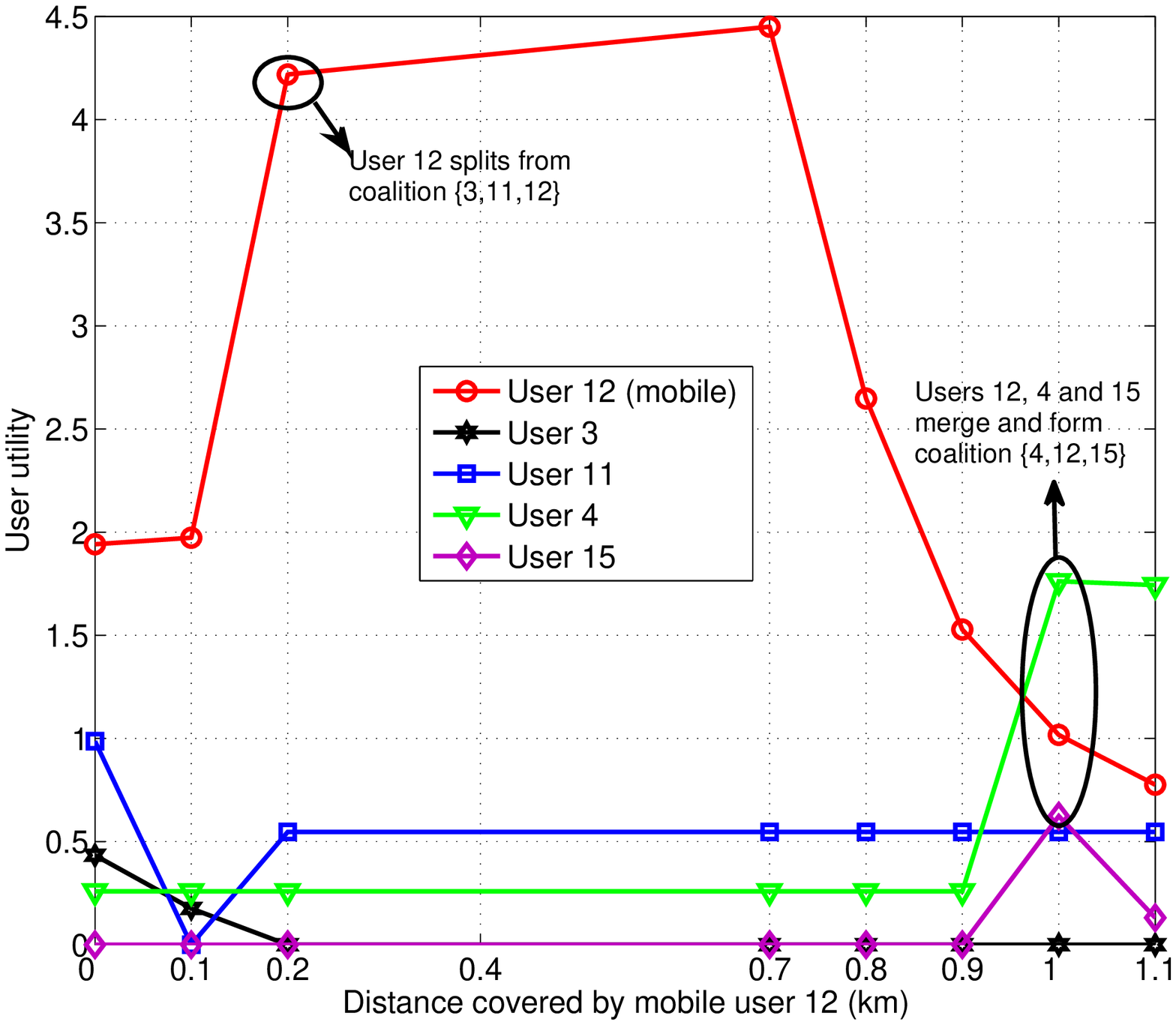}
\caption{Self-adaptation of the network's topology to mobility as User 12 in Fig.~\ref{fig:snapshot} moves horizontally on the negative x-axis (for DF).\vspace{0.52cm}}
\label{fig:mob}\vspace{-0.6cm}
\end{minipage}
\end{figure}
In Fig.~\ref{fig:snapshot}, we show a snapshot of the network structure resulting from the proposed coalition formation algorithm for a randomly deployed network with $N=15$~users and $K=2$ eavesdroppers for both DF (dashed lines) and AF (solid lines) protocols. For DF, the users self-organized into $6$ coalitions with the size of each coalition strictly larger than $K$ or equal to $1$. For example, Users~4 and 15, having no suitable partners for forming a coalition of size larger than $2$, do not cooperate. The coalition formation process is a result of Pareto order agreements for merge (or split) between the users. For example, in DF, coalition $\{5,8,10,13\}$ formed since all the users agree on its formation due to the fact that $V(\{5,8,10,13\})=\{\boldsymbol{\phi}(\{5,8,10,13\})=[0.356\ 0.8952\ 1.7235\ 0.6213]\}$ which is a clear improvement on the non-cooperative utility which was $0$ for all four users (due to proximity to eavesdropper~2). For AF, Fig.~\ref{fig:snapshot} shows that only users $\{5,8,13\}$ and users $\{1,6,7,10\}$ cooperate while all others remain non-cooperative. The main reason is that, in AF, the users need to amplify a noisy version of the signal using the beamforming weights. As a consequence, the noise can be highly amplified, and, for AF, cooperation is only beneficial in very favorable conditions. For example, coalitions $\{5,8,13\}$ and $\{1,6,7,10\}$ have formed for AF due to being far from the eavesdroppers (relatively to the other users), hence, having a small cost for information exchange. In contrast, for coalitions such as $\{3,11,12\}$, the benefit from cooperation using AF is small compared to the cost, and, thus, these coalitions do not form.

In Fig.~\ref{fig:mob} we show how the algorithm handles
mobility through appropriate coalition formation decisions. For this purpose, the network setup of Fig.~\ref{fig:snapshot}  is considered for the DF case while User $12$ is moving horizontally for $1.1$~km in the direction of the \emph{negative} x-axis. First of all, User $12$ starts getting closer to its receiver (destination~2), and, hence, it improves its utility. In the meantime, the utilities of User $12$'s partners (Users $3$ and $11$) drop due to the increasing cost. As long as the distance covered by User $12$ is less than $0.2$~km, the coalition of Users $3$, $11$ and $12$ can still bring mutual benefits to all three users. After that, splitting occurs by a mutual agreement and all three users transmit independently. When User $12$ moves about $0.8$~km, it begins to distance itself from its receiver and its utility begins to decrease. When the distance covered by User $12$ reaches about $1$~km, it will be beneficial to Users $12$, $4$, and $15$ to form a 3-user coalition through the merge rule since they improve their utilities from $\phi_4(\{4\})=0.2577$, $\phi_{12}(\{12\})=0.7638$, and $\phi_{15}(\{15\})=0$ in a non-cooperative manner to  $V(\{4,12,15\})=\{\boldsymbol{\phi}(\{4,12,15\})=[1.7618\ 1.0169\ 0.6227]\}$.

In Fig.~\ref{fig:perf} we show the performance, in terms of average utility (secrecy rate) per user, as a function of the network size $N$ for both the DF and AF cases for a network with  $K=2$~eavesdroppers. First, we note that the performance of coalition formation with DF is increasing with the size of the network, while the non-cooperative and the AF case present an almost constant performance. For instance, for the DF case, Fig.~\ref{fig:perf} shows that, by forming coalitions, the average individual utility (secrecy rate) per user is increased at all network sizes with the performance advantage of DF increasing with the network size and reaching up to $25.3\%$ and $24.4\%$ improvement over the non-cooperative and the AF cases, respectively, at $N= 45$. This is interpreted by the fact that, as the number of users $N$ increases, the probability of finding candidate partners to form coalitions with, using DF, increases for every user. Moreover, Fig.~\ref{fig:perf} shows that the performance of AF cooperation is comparable to the non-cooperative case. Hence, although AF relaying can improve the secrecy rate of large clusters of nearby cooperating users when no cost is accounted for such as in \cite{WI07}, in a practical wireless network and in the presence of a cooperation cost, the possibility of cooperation using AF for secrecy rate improvement is rare as demonstrated in Fig.~\ref{fig:perf}. This is mainly due to the strong dependence of the secrecy rate for AF cooperation on the channel between the users as per (\ref{eq:seccoopAF}), as well as the fact that, for AF, unless highly favorable conditions exist (e.g. for coalitions such as $\{1,6,7,10\}$ in Fig.~\ref{fig:snapshot}) , the amplification of the noise resulting from beamforming using AF relaying hinders the gains from cooperation relative to the secrecy cost during the information exchange phase.

\begin{figure}[!t]
\begin{minipage}[b]{0.45\linewidth}
\centering
\includegraphics[width=\textwidth]{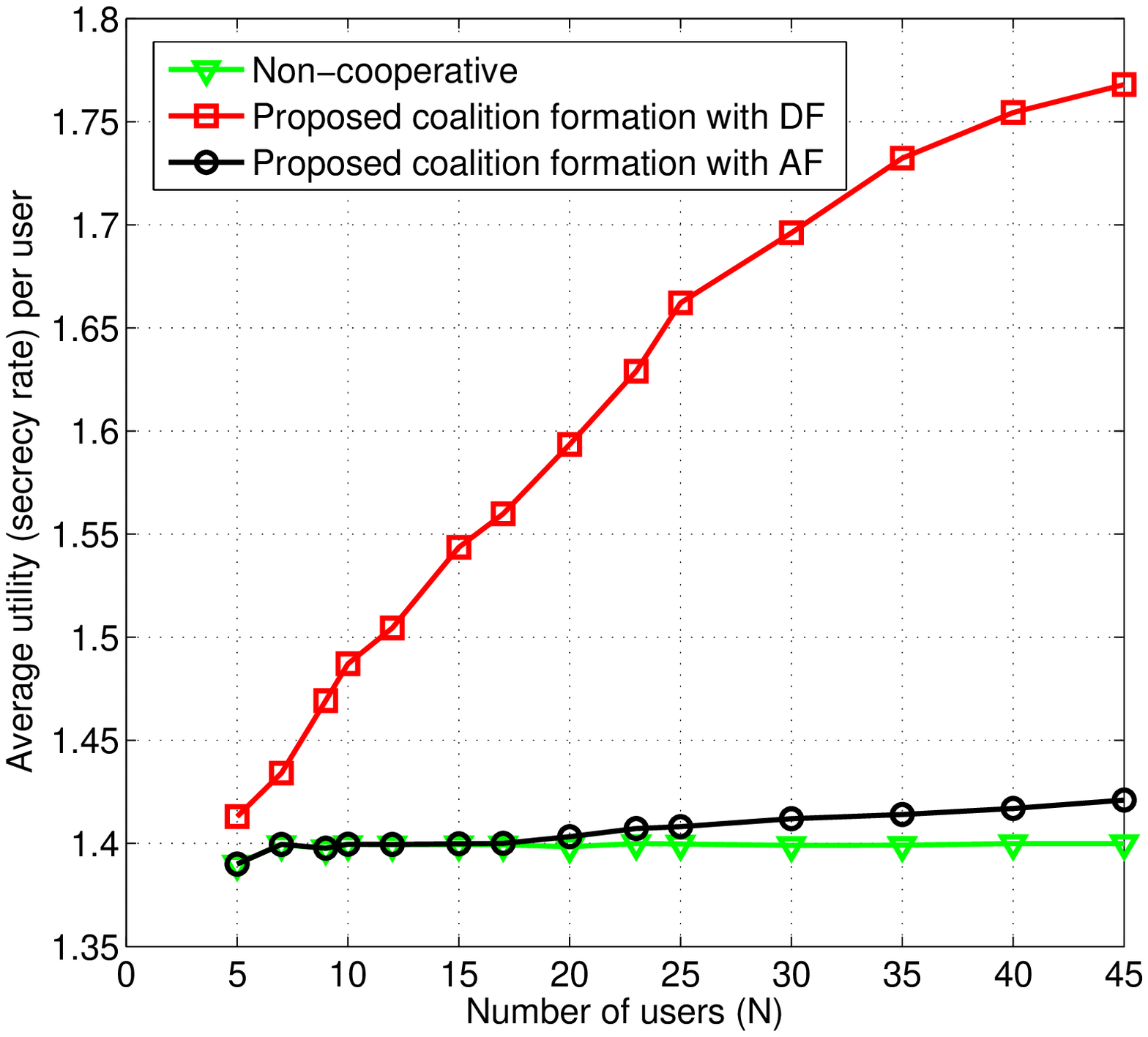}
\vspace{-0.6cm}
\caption{Performance in terms of the average individual user utility (secrecy rate) as a function of the network size $N$ for $M=2$~destinations and $K=2$~eavesdroppers.}\label{fig:perf}\vspace{-0.6cm}
\end{minipage}
\hspace{0.5cm}
\begin{minipage}[b]{0.45\linewidth}
\centering
\includegraphics[width=\textwidth]{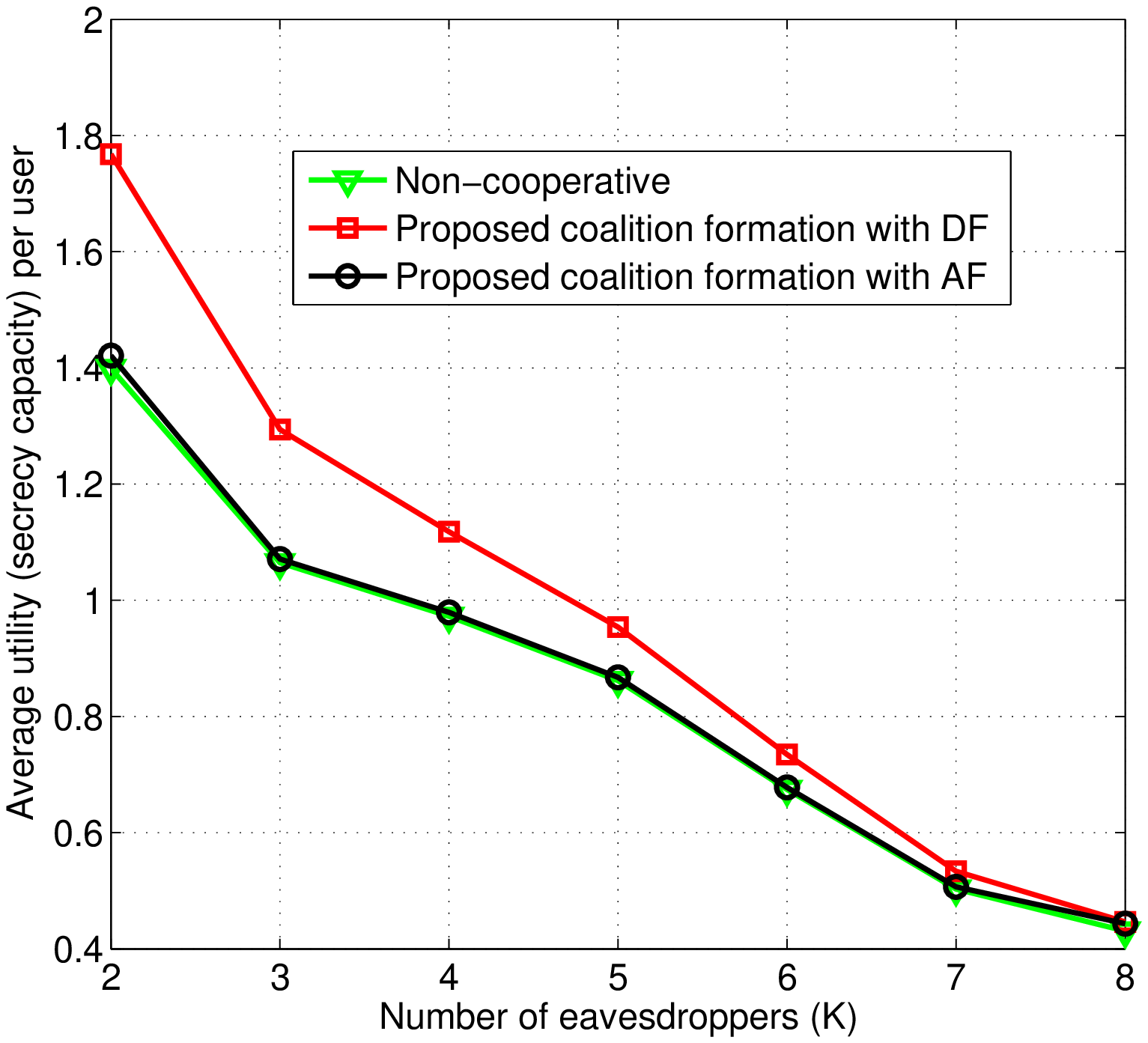}
\vspace{-0.6cm}
\caption{Performance in terms of the average individual user utility (secrecy rate) as a function of the number of eavesdroppers $K$ for $N=45$~users and $M=2$~destinations.\vspace{0.2cm}}
\label{fig:perfe}\vspace{-0.6cm}
\end{minipage}
\end{figure}

In Fig.~\ref{fig:perfe}, we show the performance, in terms of average utility (secrecy rate) per user, as the number of eavesdroppers $K$ increases for both the DF and AF cases for a network with  $N=45$~users. Fig.~\ref{fig:perf} shows that, for DF, AF and the non-cooperative case, the average secrecy rate per user decreases as more eavesdroppers are present in the area.  Moreover, for DF, the proposed coalition formation algorithm presents a performance advantage over both the non-cooperative case and the AF case at all $K$. Nonetheless, as shown by Fig.\ref{fig:perfe}, as the number of eavesdroppers increases, it becomes quite difficult for the users to improve their secrecy rate through coalition formation; consequently, at $K=8$, all three schemes exhibit a similar performance. Finally, similar to the results of Fig.~\ref{fig:perf}, coalition formation using the AF cooperation protocol has a comparable performance with that of the non-cooperative case at all $K$ as seen in Fig.~\ref{fig:perfe}.

\begin{figure}[!t]
\begin{minipage}[b]{0.45\linewidth}
\centering
\includegraphics[width=\textwidth]{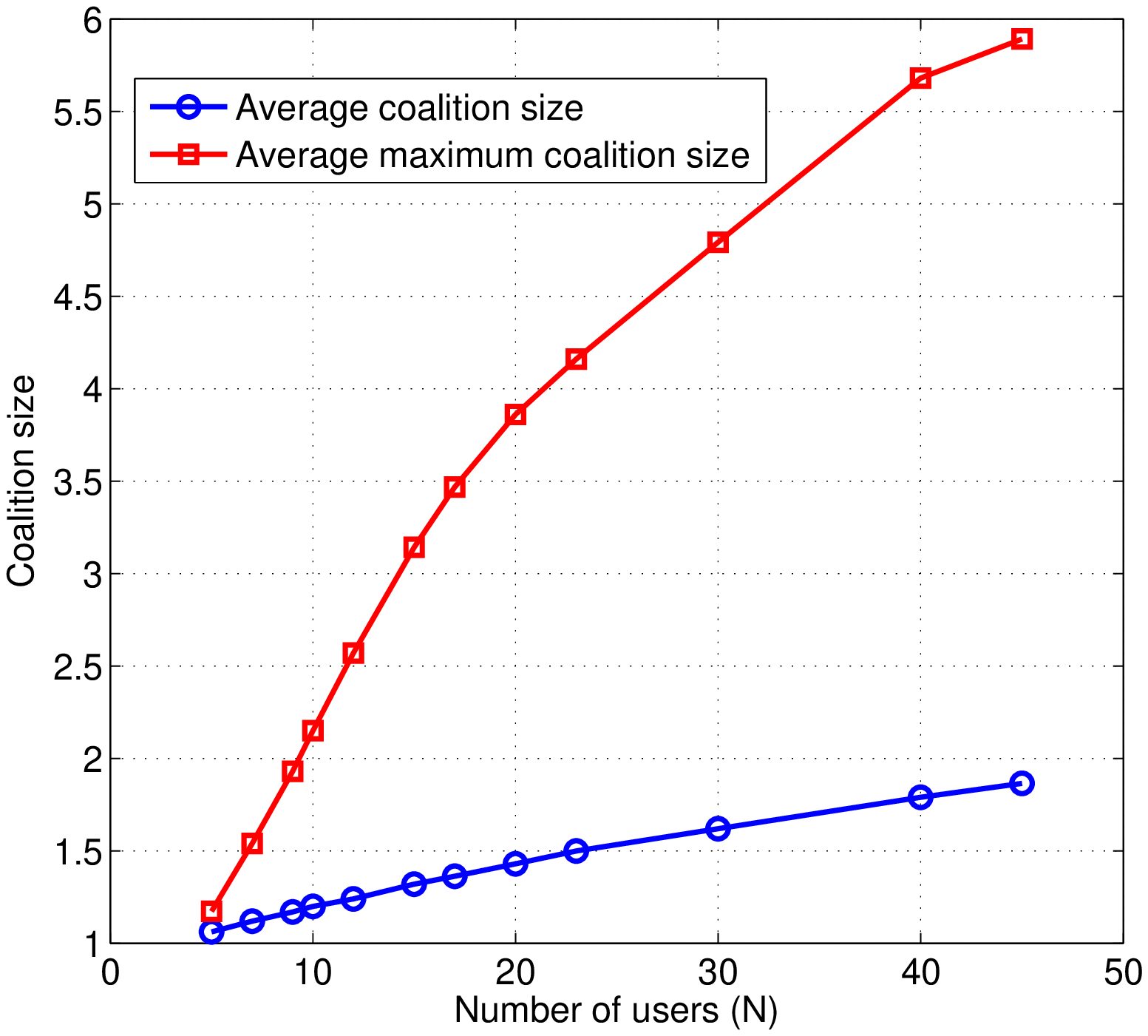}
\vspace{-0.6cm}
\caption{Average and average maximum coalition size as the network size $N$ varies for $M=2$~destinations and $K=2$~eavesdroppers and DF cooperation.}\label{fig:size}\vspace{-0.6cm}
\end{minipage}
\hspace{0.5cm}
\begin{minipage}[b]{0.45\linewidth}
\centering
\includegraphics[width=\textwidth]{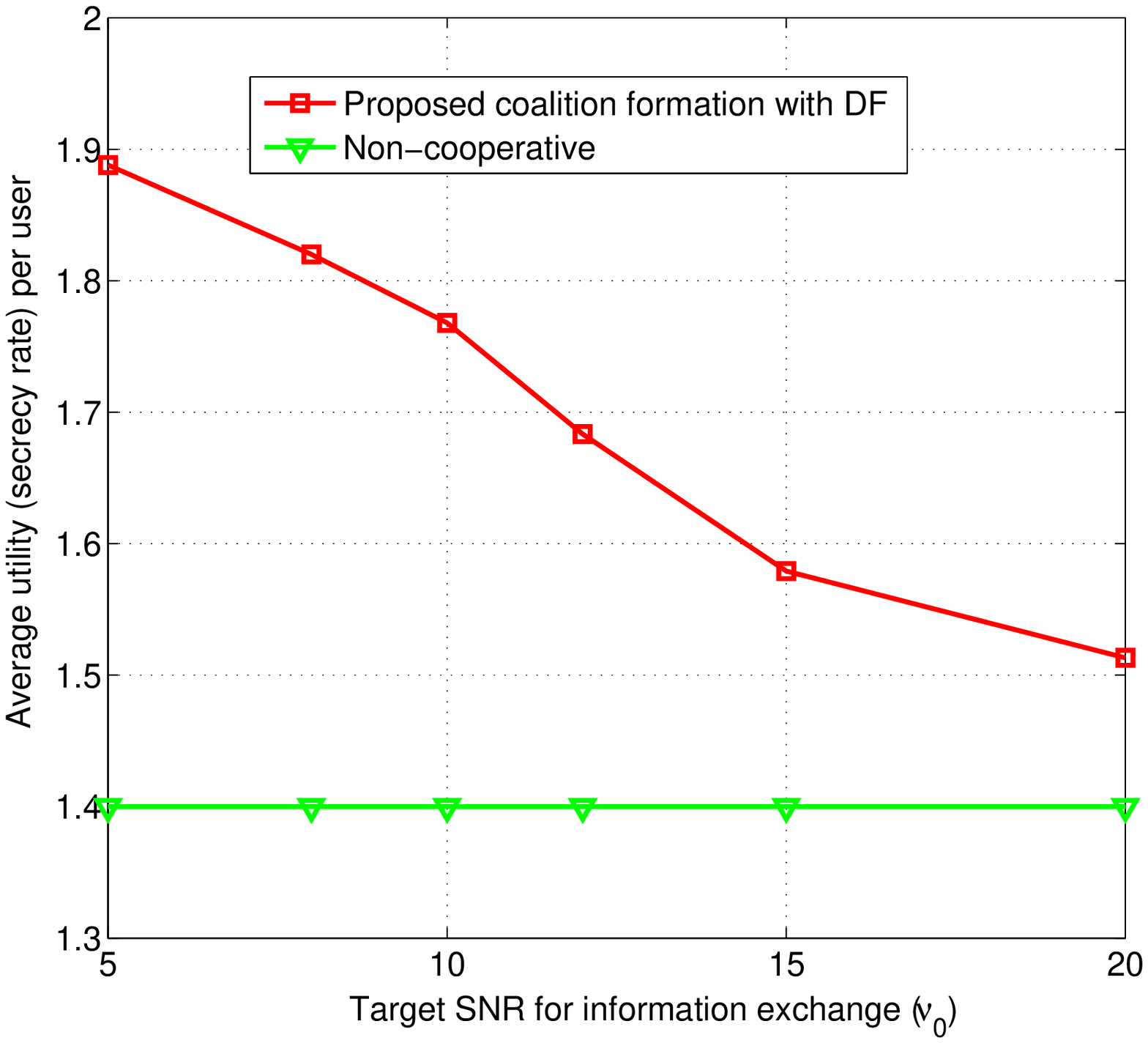}
\vspace{-0.6cm}
\caption{Average individual user utility as a function of the target SNR $\nu_0$ for information exchange for a network with $N=45$~users, $K=2$~eavesdroppers and $M=2$~destinations for DF.}
\label{fig:snr}\vspace{-0.6cm}
\end{minipage}
\end{figure}

In Fig.~\ref{fig:size}, for DF cooperation, we show the average and average maximum coalition size resulting from the proposed algorithm as the number of users, $N$, increases, for a network with $K=2$~eavesdroppers. Fig.~\ref{fig:size} shows that  both the average and average maximum coalition size increase with the number of users. This is mainly due to the fact that as $N$ increases, the number of candidate cooperating partners increases. Further, through Fig.~\ref{fig:size} we note that the formed coalitions have a small average size and a relatively large maximum size reaching up to around $2$ and $6$, respectively, at $N=45$. Since the average coalition size is below the minimum of $3$ (as per Remark~1 due to having $2$ eavesdroppers) and the average maximum coalition size is relatively large, the network structure is thus composed of a number of large coalitions with a few non-cooperative users.

In Fig.~\ref{fig:snr}, the performance, in terms of average utility (secrecy rate) per user, of the network for different cooperation costs, i.e., target average SNRs $\nu_0$ is assessed.  Fig.~\ref{fig:snr}  shows that cooperation through coalition formation with DF maintains gains, in terms of average secrecy rate per user, at almost all costs (all SNR values). However, as the cost increases and the required target SNR becomes more stringent these gains decrease converging further towards the non-cooperative gains at high cost since cooperation becomes
difficult due to the cost. As seen in Fig.~\ref{fig:snr}, the secrecy rate gains resulting from the proposed coalition formation algorithm range from $8.1\%$ at $\nu_0=20$~dB to around $34.9\%$ at $\nu_0=5$~dB improvement relative to the non-cooperative case.

The proposed algorithm's performance is further investigated in networks with $N = 20$ and $N=45$~\emph{mobile} users (random walk mobility) for a period of
$5$ minutes in the presence of $K=2$~stationary eavesdroppers. During this period, the proposed algorithm is run periodically every $30$ seconds. The results in terms of the frequency of
merge and split operations per minute are shown in Fig.~\ref{fig:freq} for various speeds. As the speed increases, the frequency of both merge and split operations per minute increases due to the changes in the network structure incurred by the increased mobility. These frequencies reach up to around $19$ merge operations per minute and $9$ split operations per minute for $N=45$ at a speed of $72$~km/h. Finally, Fig.~\ref{fig:freq} demonstrates that the frequency of merge and split operations increases with the network size $N$ as the users become more apt to finding new cooperation partners when moving which results in an increased coalition formation activity.

\begin{figure}[!t]
\begin{minipage}[b]{0.45\linewidth}
\centering
\includegraphics[width=\textwidth]{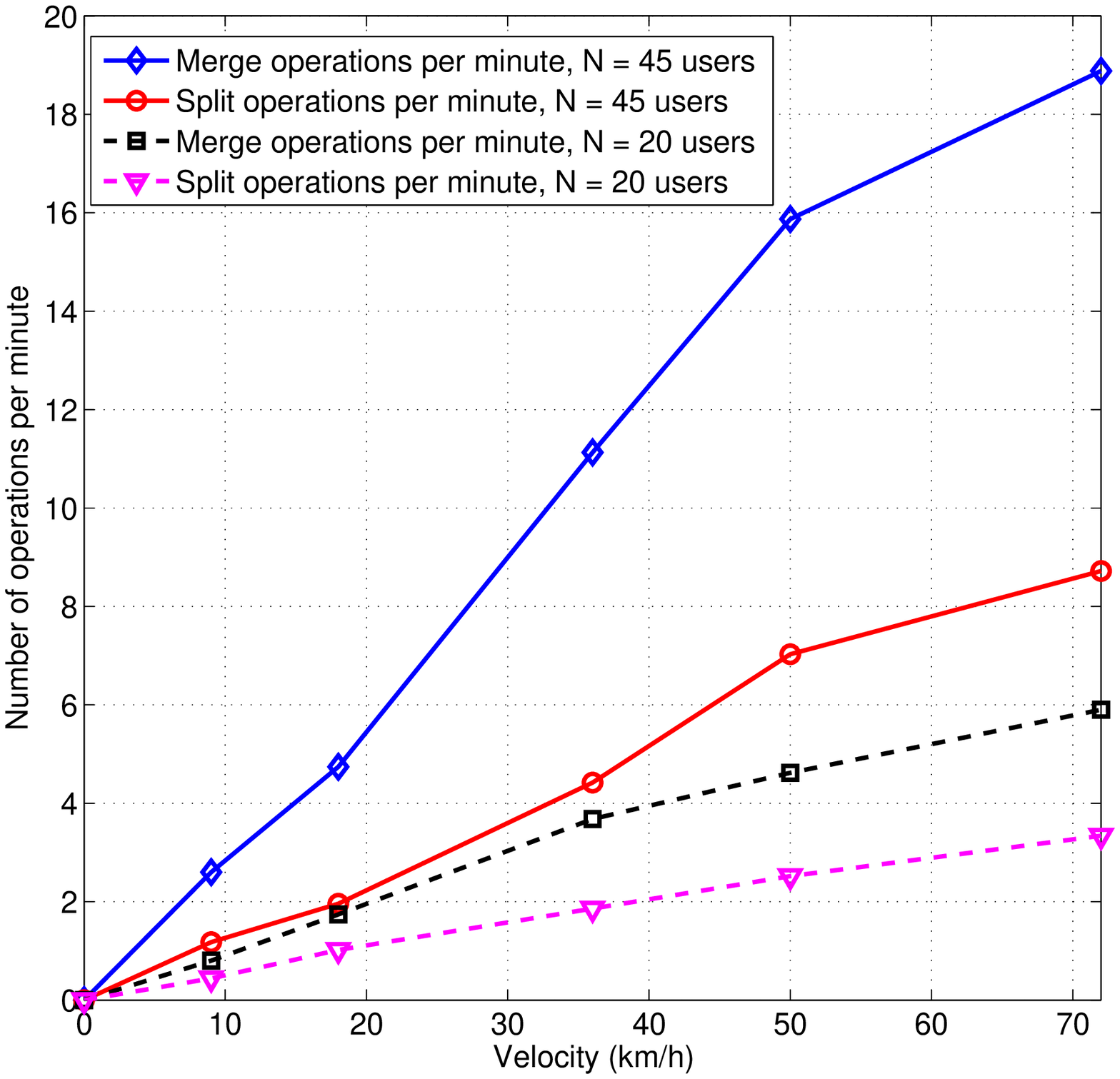}
\vspace{-0.6cm}
\caption{Frequency of merge and split operations per minute vs. speed of the users for different network sizes and $K=2$~eavesdroppers with DF cooperation.}\label{fig:freq}\vspace{-0.6cm}
\end{minipage}
\hspace{0.5cm}
\begin{minipage}[b]{0.45\linewidth}
\centering
\includegraphics[width=\textwidth]{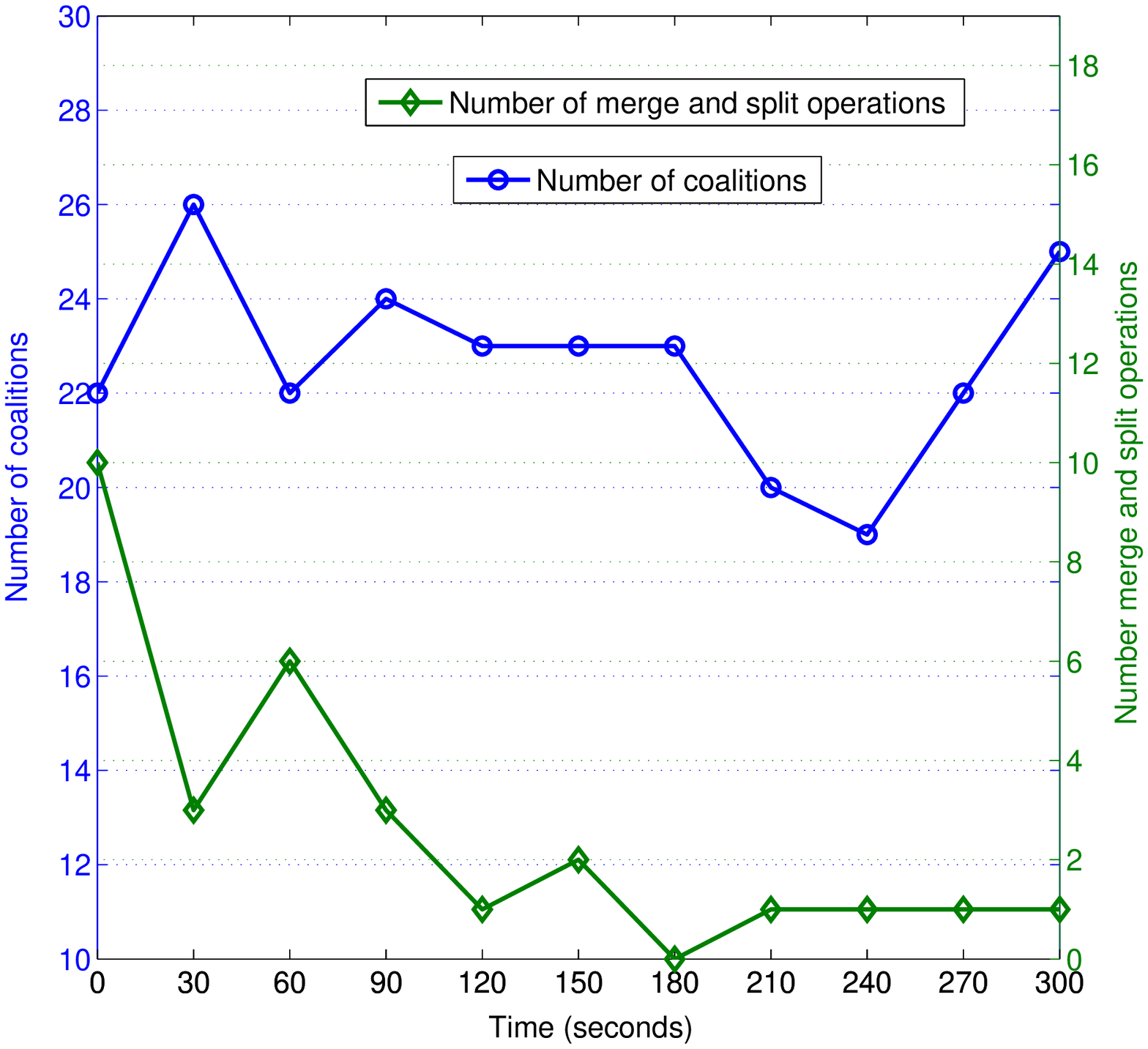}
\vspace{-0.6cm}
\caption{Evolution over time for a network with $N=45$~users, $M=2$~destinations, and $K=2$~eavesdroppers with DF cooperation when the eavesdroppers are moving with a speed of $50$~km/h.}
\label{fig:mobconv}\vspace{-0.6cm}
\end{minipage}
\end{figure}

Fig.~\ref{fig:mobconv} shows, for DF, how the structure of the wireless network with $N=45$~users and $K=2$~mobile eavesdroppers evolves and self-adapts over time (a period of $5$ minutes), while both eavesdroppers are mobile with a constant velocity of $50$~km/h. The proposed coalition formation algorithm is repeated periodically by the users every $30$ seconds, in order to provide self-adaptation to mobility. First, the users self-organize into $22$~coalitions after the occurrence of $10$ merge and split operations at time $t=0$. As time evolves, through adequate merge and split operations the network structure is adapted to the mobility of eavesdroppers. For example, at time $t=1$~minute, through a total of $6$ operations constituted of $5$ merge and $1$ split, the network structure changes from a partition of $26$~coalitions back to a partition of $22$ coalitions. Further, at $t=3$~minutes, no merge or split operations occur, and, thus, the network structure remain unchanged. In summary, Fig.~\ref{fig:mobconv} illustrates how the users can take adequate merge or split decisions to adapt the network structure to the mobility of the eavesdroppers.\vspace{-0.75cm}

\section{Conclusions}\vspace{-0.3cm}
\label{sec:conc}
In this paper, we have studied the user behavior, topology, and dynamics of a network of users that interact in order to improve their secrecy rate through both decode-and-forward and amplify-and-forward cooperation. We formulated the problem as a non-transferable coalitional game, and proposed a distributed and adaptive coalition formation algorithm. Through the proposed algorithm, the mobile users are able to take a distributed decision to form or break cooperative coalitions through well suited rules from cooperative games while maximizing their secrecy rate taking into account various costs for information exchange. We have characterized the network structure resulting from the proposed algorithm, studied its stability, and analyzed the self-adaptation of the topology to environmental changes such as mobility. Simulation results have shown that, for decode-and-forward, the proposed algorithm allowed the users to self-organize while improving the average secrecy rate per user up to $25.3\%$ and $24.4\%$ (for a network with $45$~users) relative to the non-cooperative and amplify-and-forward cases, respectively.\vspace{-0.7cm}
\def\baselinestretch{0.88}
\bibliographystyle{IEEEtran}
\bibliography{references}

\end{document}